\documentclass[journal]{IEEEtran}
\usepackage{amsmath}
\usepackage{algorithm}
\usepackage{algpseudocode}
\usepackage{leftindex}
\usepackage{amssymb}
\usepackage{algpseudocode}
\usepackage{subcaption}
\usepackage{graphicx}
\usepackage{amsthm}
\usepackage[justification=centering]{caption}

\theoremstyle{definition}

\title{Data-Driven LQR using Reinforcement Learning and Quadratic Neural Networks}
\author{Soroush Asri, Luis Rodrigues\\Department of Electrical Engineering, Concordia University}

\newtheorem{remark}{Remark}

\newtheorem{lemma}{Lemma}

\begin{document}

\maketitle

\begin{abstract}

%This article designs linear quadratic regulators (LQR) by Q-learning, a class of reinforcement learning algorithms, where the policy evaluator is chosen as a two-layer quadratic neural network (QNN) trained by a convex optimization. In this approach the system model is not required.  

This paper introduces a novel data-driven approach to design a linear quadratic regulator (LQR) using a reinforcement learning (RL) algorithm that does not require the system model. The key contribution is to perform policy iteration (PI) by designing the policy evaluator as a two-layer quadratic neural network (QNN).
This network is trained through convex optimization. To the best of our knowledge, this is  the first time that a QNN trained through convex optimization is employed as the Q-function approximator (QFA).
The main advantage is that the QNN's input-output mapping has an analytical expression as a quadratic form, which can then be used to obtain an analytical expression for policy improvement.
This is in stark contrast to available techniques in the literature that must train a second neural network to obtain the policy improvement.
The article establishes the convergence of the learning algorithm to the optimal control provided the system is controllable and one starts from a stabilitzing policy.
A quadrotor example demonstrates the effectiveness of the proposed approach.

%This is the first time  a QNN trained by a convex optimization is chosen as the Q-function approximator (QFA). The quadratic input-output mapping of the QNN is utilized to perform the policy improvement step analytically. Since the value-function of LQR is also quadratic, QNN is a good candidate. The convergence to the optimal control is proven and several examples show the effectiveness of this novel method.

    %\begin{enumerate}
    %    \item Obtaining globally optimal weights.
    %    \item Solving the SDP problem is computationally efficient because SDPs can be solved in polynomial time.
    %    \item Policy improvement step can be done analytically since QNN has a quadratic input-output mapping.
    %    \item The optimal number of neurons in the hidden layer is obtained by solving the SDP.
    %\end{enumerate}
    
    %Further studies will address designing optimal controllers for nonlinear systems. 
\end{abstract}

\begin{IEEEkeywords}
LQR, Reinforcement Learning, Policy Iteration, Quadratic Neural Networks.
\end{IEEEkeywords}

\section{Introduction}

Optimal controllers minimize a given cost function subject to a dynamic model.
This design method has received a lot of attention, especially due to its potential in applications such as autonomous vehicle navigation and robotics, economics and management, as well as energy optimization~\cite{bryson2018applied}, to cite a few.
Traditionally, optimal control design relies on well-established techniques like dynamic programming, calculus of variations, and Pontryagin's maximum principle~\cite{lewisoptimalbook}. While these methods have proven effective in many cases, they often face challenges when dealing with large-scale, nonlinear, or uncertain systems~\cite{extendedARE}\cite{jiang2015global}. To address this issue, one often approximates a nonlinear system by a linear model around an operating point, thus enabling the design of a linear quadratic regulator (LQR)~\cite{kalman1960contributions}.
One of the main advantages of LQR~\cite{kalman1960contributions} is that it has a closed-form solution~\cite{anderson2007optimal}\cite{sima2021algorithms} as a function of a matrix that is the positive definite solution of a Riccati equation~\cite{lewisoptimalbook}. 
An iterative technique for solving the Riccati equation is discussed in article~\cite{kleinman1968iterative}, which requires an accurate linear model of the system.
However, there are instances where complex nonlinear systems cannot be adequately approximated by linear models or where the system's dynamics remain unknown. 
Therefore, conventional techniques for solving optimal control problems may not be applicable. As a result, researchers have explored data-driven methods such as RL, imitation learning, and Gaussian process regression as alternatives for solving optimal control problems~\cite{tailor2019learning}\cite{proctor2016dynamic}\cite{boedecker2014approximate}.
%It is known that one can obtain an optimal controller with quadratic performance index for linear and nonlinear systems  by solving the algebraic Riccati equation (ARE) and the Hamilton–Jacobi–Bellman equation (HJB), respectively~\cite{lewisoptimalbook}. In order to solve the ARE to design linear-quadratic-regulators (LQRs), article~\cite{kleinman1968iterative} proposed an iterative technique with the assumption that the system model is known. However, in many cases the system model is unknown or inaccurate~\citep{extendedARE, jiang2015global}. Additionally, for nonlinear systems solving the HJB equation can be difficult or even impossible~\cite{kiumarsi2017optimal}.
RL is a powerful data-driven approach that enables an agent to learn optimal control policies by interacting with the environment and adjusting its actions or control policies  to minimize a cumulative cost~\cite{kaelbling1996reinforcement}. The motivation to use RL as an approxiation to optimal control designs stems from several key factors, including adaptability, the ability to handle complex dynamics, learning from real-world experience, and overcoming inaccuracies in system models~\cite{bertsekas2019reinforcement}\cite{kiumarsi2017optimal}. 
The majority of research dedicated to RL in optimal control investigates discrete-time optimal control using value-based RL algorithms~\cite{lewis2009reinforcement}.
In value-based RL, the core concept is to estimate value functions, such as the Q-function~\cite{bertsekas2019reinforcement}, which can then be used to derive optimal policies. A common approach to estimate value functions  is to use the temporal difference (TD) equation~\cite{TD2018}.
In such approach the Bellman error~\cite{tedrake2016underactuated} is reduced by iteratively updating the value function approximation. This is called adaptive dynamic programming (ADP)~\cite{lewis2009reinforcement}. 
ADP is widely favored due to its link with dynamic programming and the Bellman equation~\cite{bertsekas2012dynamic}. The application of ADP methods such as the policy iteration (PI) and the value-iteration (VI) to feedback control systems are discussed in references~\cite{lewis2009reinforcement}\cite{prokhorov1997adaptive}. 
For optimal control, PI is preferred over VI  due to its stability arising from the policy improvement guarantee, which leads to reliable convergence to the optimal policy~\cite{LQR_RL_stable}.  
Additionally, Q-learning, an ADP algorithm that approximates the Q-function, can be used to achieve optimal control without relying on a system model~\cite{sutton2018reinforcement}.

It is common to use a neural network (NN) to approximate the value function during the policy evaluation step of a PI algorithm~\cite{khan2012reinforcement}\cite{ADHDP}.  
However, the drawbacks of using NNs as the value function approximator (VFA) are
    %\item The optimization problem to approximate the Q-function with neural networks is not convex. As a result, we might acquire a local optimum solution for the QFA.
    (i) the optimization of the neural network's weights is not convex and training the neural network yields only locally optimal weights,
    %is not the the most accurate approximation with the given structure.
    (ii) selecting the appropriate architecture for the NN typically involves a trial and error procedure, 
    (iii) the input-output mapping of the NN lacks an analytical expression and as a consequence a second neural network is needed to perform policy improvement,
    (iv) providing a comprehensive proof of convergence to the optimal control is difficult or even unattainable.
To address the mentioned issues, a two-layer quadratic neural network (QNN) trained by a convex optimization introduced in reference~\cite{QNN} will be chosen as the VFA. The advantages of using the two-layer QNN as the VFA compared to other neural networks are:
    (i) Two-layer QNNs are trained by solving a convex optimization and therefore the global optimal weights are found~\cite{QNN},
    (ii) the optimal number of neurons in the hidden layer is obtained ad a by-product~\cite{QNN},
    (iii) the input-output mapping of the QNN is a quadratic form~\cite{QNN,Luis_QNN} and therefore one can analytically minimize the value-function with respect to the  control policy.
Reference~\cite{Luis_QNN} addresses applications of QNNs to system identification and control of dynamical systems.

LQR design problems have been addressed using various RL approaches as they can serve as benchmarks for evaluating the performance of RL algorithms~\cite{matni2019self,yaghmaie2022linear} because their solution is known in closed-form. 
This fact has motivated the use of ADP to solve discrete-time LQR in reference~\cite{bradtke1994adaptive}, discrete-time LQ Gaussian in references~\cite{lewis2010reinforcement}\cite{rizvi2019experience}, and discrete-time LQ tracking in references~\cite{kiumarsi2014reinforcement}\cite{optimal_tracking_with_measured_input_output} without using the system model.     
Moreover, in the Q-learning scheme proposed in reference~\cite{bradtke1994adaptive} one can obtain the desired LQR controller by solving a least-squares optimization to derive the Q-function provided the system
is controllable and a stabilizing linear state feedback controller is known.

This paper designs a discrete-time LQR controller using Q-learning.
More specifically, we propose to use a two-layer QNN as the VFA in Q-learning for a LQR problem, for which the value function is known to be quadratic.
To the best of our knowledge, this represents the first result with a convex optimization-trained QNN used as a VFA. Additionally, the designed controller is proven to converge to the LQR controller provided the system is controllable and an initial stabilizing policy is known.
Simulations conducted in MATLAB illustrate the convergence of the learning algorithm to the optimal controller for different initial stabilizing policies.

This paper is organized as follows. Section~\ref{section: problem statement} presents the problem statement. Section \ref{section: QNN} reviews two-layer QNNs. Section~\ref{section: policy-iteration based ADHDP} presents the general form of the Q-learning algorithm.  Section~\ref{section: our approach} contains the proposed approach with the proof of convergence. Section~\ref{section: LQR implementation} focuses on examples, which is followed by the conclusions.

\section{Problem statement}
\label{section: problem statement}
An unknown controllable linear system model is written as
    \begin{gather} 
          x_{k+1}=Ax_k+Bu_k \label{linear_dynamics}
    \end{gather} 
    
    \noindent where $x_k \in \mathbb{R}^{n_x}$ is the state vector, $u_k \in \mathbb{R}^{n_u}$ is the input vector and $A, \:  B $ are unknown matrices constrained to be such that the pair $(A,B)$ is controllable.
Define the policy $\pi (.)$ as a linear mapping from the state vector to the input vector as
    \begin{gather} 
          u_{k}=\pi (x_k) = -K^\pi x_k 
    \end{gather} 
    
    \noindent where $K^\pi$ is a matrix to be determined by the designer. Define the local cost as
    \begin{gather}
        c(x_k,u_k)= x_k^T Q x_k + u_k^T R u_k
    \end{gather}

    \noindent where $Q$, and $R$ are both positive definite. The value function when one follows the control policy $\pi(.)$ is defined as
    \begin{equation}
        V^\pi(x_k) =   \sum_{i=k}^{\infty} \gamma^{i-k} \left( x_i^T Q x_i + u_i^T R u_i \right)
    \end{equation}

 \noindent where $0<\gamma \leq 1$ is a discount factor.

The objective is to obtain the optimal policy $\pi^* (.)$ that minimizes $V^{\pi}(x_k)$ for all states $x_k$ subject to the unknown dynamics of the system.
This is done using QNNs as a VFA.

% \begin{definition}[Admissible policy~\cite{ADHDP}]
% \label{def: Admissible policy}
% The policy $\pi (.)$ is admissible if
% \begin{enumerate}
%    \item The system reaches the origin following $\pi (.)$ from any initial state.
%    \item for all $ \Vert x_k \Vert < \infty $ we have  $V^\pi (x_k) < \infty  $.
% \end{enumerate}
% \end{definition}

\begin{remark}
 If $\gamma =1 $ then the optimal control problem is a LQR problem.
\end{remark}

%\begin{remark}
% if $ 0< \gamma < 1$ is chosen, we can ensure the second condition of the definition~\ref{def: Admissible policy}  is met if the first condition is satisfied. The proof is as follows.

%\noindent If the first condition is satisfied and $ \Vert x_k \Vert < \infty $, the local cost $c(x_i,\pi (x_i))$ is finite for all $i \geq k$. We can write $V^{\pi}(x_k) $ as

%\begin{gather}
%    V^{\pi}(x_k) =  \sum^{N}_{i=0} \gamma^i c(x_{k+i}, \pi(x_{k+i})      ) + \sum^{\infty}_{i=N+1} \gamma^i c(x_{k+i}, \pi(x_{k+i})      ) 
%\end{gather}

%\noindent If $N$ is chosen large enough, $\forall i> N$ we have $\gamma ^i \approx  0 $ . Therefore, for a large $N$ we have  

%\begin{gather}
%    V^{\pi}(x_k) \approx  \sum^{N}_{i=0} \gamma^i c(x_{k+i}, \pi(x_{k+i})      )
%\end{gather}

%\noindent and we know that the summation of finite numbers of finite values is finite. Therefore, we can conclude the second condition is also true.

%\end{remark}

\section{Two-layer QNNs}
\label{section: QNN}
Consider the neural network in Fig~\ref{fig: QNN_with_one_output}
with one hidden layer, one output, and a degree two polynomial activation function, where  $ \mathcal{X}_i \in \mathbb{R}^n$ is the $i$-th input data vector, $\mathcal{\hat{Y}}_i \in \mathbb{R} $ is the corresponding output, $\mathcal{Y}_i \in \mathbb{R}$ is the output label corresponding to  the input $\mathcal{X}_i$, $L$ is the number of hidden neurons, $f(z)=az^2+bz+c$ is the polynomial activation function,  and $a \neq 0$, $b$, and $c$ are pre-defined constant coefficients. The notation $w_{kj}$ represents the weight from the $k$-th input-neuron  to the $j$-th hidden-neuron, and $\rho_{j}$ represents the weight from the j-th hidden-neuron to the output.
The input-output mapping is
\begin{figure}[t]
\begin{center}
\includegraphics[width=3.5in]{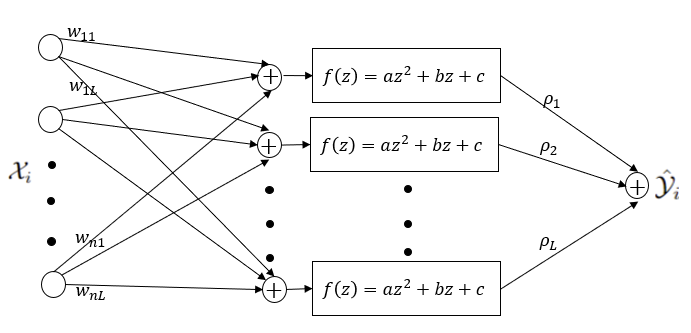}
\caption{two-layer QNN with one output}
\label{fig: QNN_with_one_output}
\end{center}
\end{figure}
%It is suggested in~\cite{QNN} to choose $a=0.0937, b=0.5, c=0.4688$ to approximate the ReLU function in the sense of least-squares in the interval [-5, 5]. 
      \begin{equation}
         \mathcal{\hat{Y}}_i=\sum_{j=1}^{L} f(\mathcal{X}_i^TW_j)\rho_{j} 
      \end{equation}    
      \noindent where $W_j = \begin{bmatrix}
        w_{1j} &
        w_{2j} &
        \hdots &
        w_{nj}
      \end{bmatrix}^T$.      
      %To train the neural network, the following optimization problem should be solved.
%\begin{equation}
%    \begin{split} \label{eq: normal neural net optimizaion}
%      \min_{W_k,v_k} \: \: & l(\hat{Y}-Y)\\
%        s.t. \: \: & \hat{Y}_i=\sum_{j=1}^{L} f(X_i^TW_j)v_{j} \\
%                   & k=1,2,...,L \: \: \: \: \: \:i=1,2,...,N
%    \end{split}
%\end{equation}

Reference~\cite{QNN} proposes the training optimization 
%     \begin{equation}
%     \begin{split} \label{eq: changed neural net optimizaion}
%       \min_{W_i,v_i} \: \: & l(\hat{y}-y) + \beta \sum_{j=1}^{L} | v_j |\\
%       s.t. \: \: & \hat{y}=\sum_{j=1}^{L} f(\hat{x}^TW_j)v_{j} \\ 
%                   & \Vert  W_i \Vert_2 =1, \: \: \: \: i=1,2, \hdots , L 
%     \end{split}
% \end{equation}
\begin{equation}
    \begin{split} \label{eq: changed neural net optimizaion}
      \min_{W_k,\rho_k} \: \: & l(\mathcal{\hat{Y}}-\mathcal{Y}) + \beta \sum_{j=1}^{L} | \rho_j |\\
        s.t. \: \: & \mathcal{\hat{Y}}_i=\sum_{j=1}^{L} f(\mathcal{X}_i^TW_j)\rho_{j}\\
         & \Vert  W_k \Vert_2 =1, \: \: \: \: k=1,2,...,L, \: \: \: \:i=1,2,...,N 
    \end{split}
\end{equation}

    %   \noindent It is shown in~\cite{QNN} that the dual of the optimization problem~\eqref{primal_problem} has no gap with the primal problem. The neural network of the dual problem is shown in Fig~\ref{fig: QNN_dual}.

    % \begin{figure}[htbp]
    % \begin{center}
    % \includegraphics[width=4in]{QNN_dual.png}
    % \caption{The neural network of the dual problem}
    % \label{fig: QNN_dual}
    % \end{center}
    % \end{figure}

    %   \noindent In this neural network, the output neurons are decoupled in the sense that each output neuron has its own hidden neurons.

\noindent where $\beta\ge 0$ is a pre-defined regularization parameter, $l(.)$ is a convex loss function, N is the number of data points, $\mathcal{\hat{Y}} = \begin{bmatrix}
        \mathcal{\hat{Y}}_1 &
        \mathcal{\hat{Y}}_2 &
        \hdots &
        \mathcal{\hat{Y}}_N
      \end{bmatrix}^T$, and $\mathcal{Y} = \begin{bmatrix}
        \mathcal{Y}_1 &
        \mathcal{Y}_2 &
        \hdots &
        \mathcal{Y}_N
      \end{bmatrix}^T$. 

      \vspace{2mm}
The optimization problem~\eqref{eq: changed neural net optimizaion} can be equivalently solved by the dual convex optimization~\cite{QNN}
 \begin{equation}
\begin{split} \label{convex_formula}
      \min_{Z^+ , Z^-} \: \: & l(\mathcal{\hat{Y}} - \mathcal{Y}) + \beta  \:  Trace(Z_{1}^+ + Z_{1}^-)  \\
        s.t. \: \: & \mathcal{\hat{Y}}_{i}=a\mathcal{X}_{i}^T(Z_{1}^+ - Z_{1}^-)\mathcal{X}_{i}   +  b\mathcal{X}_{i}^T(Z_{2}^+ - Z_{2}^-) + \\
        & \: \: \: \: \: \: \: \:  \: \:  c \:  Trace(Z_{1}^+ - Z_{1}^-), \\ 
          & Z^+=\begin{bmatrix} Z_{1}^+ &    Z_{2}^+ \\(Z_{2}^+)^T &Trace(Z_{1}^+) \end{bmatrix} \geq 0 , \\
          & Z^-=\begin{bmatrix} Z_{1}^- &    Z_{2}^- \\(Z_{2}^-)^T &Trace(Z_{1}^-) \end{bmatrix} \geq 0, \\
          &  i=1, 2, \hdots N 
\end{split}
\end{equation}
% \begin{equation}
% \begin{split} \label{convex_formula}
%       \min_{Z_k^+ , Z_k^-} \: \: & l(\hat{y}-y) + \beta \sum_{k=1}^{m} Trace(Z_{k,1}^+ + Z_{k,1}^-) \\
%         s.t. \: \: & \hat{y}_{k(i)}=ax_{(i)}^T(Z_{k,1}^+ - Z_{k,1}^-)x_{(i)}   +   bx_{(i)}^T(Z_{k,2}^+ - Z_{k,2}^-)+c \: Trace(Z_{k,1}^+ - Z_{k,1}^-) \\
%           & Z_k^+=\begin{bmatrix} Z_{k,1}^+ &    Z_{k,2}^+ \\Z_{k,2}^{+^T} &Trace(Z_{k,1}^+) \end{bmatrix} \geq 0 , \: \:
%           Z_k^-=\begin{bmatrix} Z_{k,1}^- &    Z_{k,2}^- \\Z_{k,2}^{-^T} &Trace(Z_{k,1}^-) \end{bmatrix} \geq 0 \\ 
%           & k=1, 2, \hdots m \: \: , \: \: i=1, 2, \hdots N
% \end{split}
% \end{equation}
where $Z_1^+$, $Z_2^+$, $Z_1^-$, $Z_2^-$ are optimization parameters~\cite{QNN}.
After training the neural network and obtaining $Z^+$, $Z^-$ from~\eqref{convex_formula}, the quadratic input-output mapping is~\cite{QNN,Luis_QNN}
    
    \begin{equation} \label{quadratic_mapping_QNN}
    \mathcal{\hat{Y}}_i=\begin{bmatrix}
    \mathcal{X}_i  \\ 1 \end{bmatrix}^T
    H
    \begin{bmatrix}
    \mathcal{X}_i  \\ 1 \end{bmatrix} %= X^T \Bar{H} X 
     \end{equation}
where
\begin{eqnarray*}
H=\begin{bmatrix} a(Z_{1}^+ - Z_{1}^-) &    0.5b(Z_{2}^+ - Z_{2}^-)  \\0.5b(Z_{2}^+ - Z_{2}^-)^T &
    c   \left [ Trace  \left (Z_{1}^+ - Z_{1}^-  \right ) \right ]
    \end{bmatrix}
\end{eqnarray*}
    %\noindent where $X=\begin{bmatrix} \hat{x} \\ 1 \end{bmatrix}$,
    %$\Bar{H}=\begin{bmatrix} a(Z_{1}^+ - Z_{1}^-) &    0.5b(Z_{2}^+ - Z_{2}^-)  \\0.5b(Z_{2}^+ - Z_{2}^-)^T &
    %c \:  Trace(Z_{1}^+ - Z_{1}^-)
    %\end{bmatrix}
    %$.

    \begin{remark}
                If $b=c=0,$ and $ a =1$ then 
        \begin{equation}
            \mathcal{\hat{Y}}_i=
            \mathcal{X}^T_i  
            (Z_{1}^+ - Z_{1}^-)
            \mathcal{X}_i
        \end{equation}
    \end{remark}

\section{policy iteration-based Q-learning}
\label{section: policy-iteration based ADHDP}

 This section presents the Q-learning algorithm used to calculate the optimal policy $\pi^* (.)$.
The Q-function~\cite{watkins1989learning} is
   \begin{equation} 
         Q^{\pi}(x_k,u_k) = c(x_k,u_k) + \gamma V^{\pi}(x_{k+1}) \label{eqn: Q-function_def}
    \end{equation}

% To ensure that the asymptotically stabilizing control policy has a finite value function, we can define the value function as 

%     \begin{gather} 
%           V^{\pi}(x_k) = \sum^{\infty}_{i=0}  \gamma^i c(x_{k+i}, \pi(x_{k+i})      )  \label{eq: value-function with gamma} \\
%           % Q^{\pi}(x_k,u_k) = c(x_k,u_k) + \gamma V^{\pi}(x_{k+1}) 
%     \end{gather} 

% \noindent where $0<\gamma<1$ is a discount factor.

% \noindent \textbf{Definition:} (Admissible policy)
% The policy $\pi$ is admissible if 
% \begin{enumerate}
%     \item The system reaches the equilibrium point following $\pi$ from any initial state 
%     \item $V^\pi (x_k) < \infty, \: \forall x_k$
%     \item $\pi(0)=0$
% \end{enumerate}

  % \noindent   The Q-function $Q^{\pi}(x_k,u_k)$ is the cost-to-go from the state $x_k$ if in the state $x_k$ the action $u_k$ is taken and the policy $\pi (.)$ is followed afterwards.
Using Bellman's principle of optimality~\cite{lewis2009reinforcement} and the definition of Q-function, equation~\eqref{eqn: Q-function_def} can be rewritten as
     \begin{gather} 
         Q^{\pi}(x_k,u_k) = c(x_k,u_k) + \gamma Q^{\pi}(x_{k+1},\pi(x_{k+1})) \label{eqn: Bellman-eqn}
    \end{gather}
    Obtaining the Q-function by equation~\eqref{eqn: Bellman-eqn} is known as the policy evaluation step~\cite{sutton2018reinforcement}.
After the policy evaluation step, the optimal stabilizing policy $\pi^\prime (.)$ for $Q^\pi$ is given by
     \begin{gather} 
         \pi^\prime (x_k) = \arg \min_{u_k} Q^{\pi}(x_k,u_k) \label{eq: policy improvement vanila}
    \end{gather} 
where $V^{\pi^\prime}(x_k) <  V^{\pi}(x_k)$ for all $x_k$~\cite{ADHDP}.
Solving \eqref{eq: policy improvement vanila} is known as the policy improvement step.
The optimal policy $\pi^* (.)$ can be attained from any initial stabilizing policy by repeatedly performing policy evaluation and policy improvement steps until convergence, starting from an initial stabilizing policy $\pi_0(.)$. Algorithm~\ref{alg: PI_ADHDP_vanila} is the PI-based Q-learning algorithm.

%Therefore, one can obtain the optimal policy $\pi^* (.)$ starting from any initial stabilizing policy and repetitively solving the policy evaluation and policy improvement step until the policy converges. 

    %  \begin{algorithm}
    %     \caption{Policy iteration-based ADHDP} \label{alg: PI_ADHDP_vanila}
    %     \begin{algorithmic} 
    %         \State Select an initial admissible policy $\pi_0(x_k) = K^{\pi_0}x_k, \: j=0$. 

    %         \While{$\Vert K^{\pi_j}  - K^{\pi_{j-1}} \Vert > \epsilon$}
    %         \State \textbf{Policy Evaluation:} 
    %         \State $\: \: \:$ Obtain $Q^{\pi_j}(.)$ by solving $Q^{\pi_j}(x_k,u_k) = c(x_k,u_k) + \gamma Q^{\pi_j}(x_{k+1},\pi_j(x_{k+1}))$
    %         \State \textbf{Policy improvement:}
    %         \State $\: \: \:$ Obtain $\pi_{j+1}(.)$ by solving $\pi_{j+1}(x_k) = \arg \min_{u_k} Q^{\pi_j}(x_k,u_k)$
    %         \State $\: \: \: \: j \gets j+1$
    %         \EndWhile
    %         \State $\pi^*(.) \gets  \pi_j(.)$ 
    %     \end{algorithmic}
    % \end{algorithm}.

         \begin{algorithm}[h]
        \caption{Policy iteration-based Q-learning} \label{alg: PI_ADHDP_vanila}
        \begin{algorithmic} 
            \State Select an initial stabilizing policy $\pi_0(x_k) = K^{\pi_0}x_k$. Then, for $j=0,1,2,\hdots$ perform policy evaluation and policy improvement steps as follows,
            \State \textbf{Policy Evaluation:} 
            \State $\: \: \:$ Obtain $Q^{\pi_j}(.)$ by solving 
            \begin{equation}
                \label{eq: PEA1}
                Q^{\pi_j}(x_k,u_k) = c(x_k,u_k) + \gamma Q^{\pi_j}(x_{k+1},\pi_j(x_{k+1}))
            \end{equation}
            \State \textbf{Policy improvement:}
            \State $\: \: \:$ Obtain $\pi_{j+1}(.)$ by solving 
            \begin{equation}
                \pi_{j+1}(x_k) = \arg \min_{u_k} Q^{\pi_j}(x_k,u_k)
            \end{equation}
        \end{algorithmic}
    \end{algorithm}
    
Solving equation~\eqref{eq: PEA1} can be challenging as it is a Lyapunov function \cite{lewis2009reinforcement}. However, exploiting its fixed-point nature, we can obtain $Q^\pi(.)$ for the stabilizing policy $\pi(.)$ as
    \begin{equation}
        \label{eq: iterative Bellman van}
        Q^{\pi}_{i+1}(x_{k},u_k) = c(x_k,u_k) +  \gamma  Q^{\pi}_{i}(x_{k+1},\pi(x_{k+1})) 
    \end{equation}
Under certain conditions that will be detailed later, starting with any initial value $Q^\pi_0$ and using the iteration (\ref{eq: iterative Bellman van}) will lead to $Q^\pi_i \rightarrow Q^\pi$. Consequently,  equation \eqref{eq: PEA1}  will be replaced by the iterative equation~\eqref{eq: iterative Bellman van}.

\begin{remark}
     PI algorithms require persistent excitation (PE)~\cite{lewis2009reinforcement}\cite{lewis2010reinforcement}. To achieve PE, a probing noise term $n_k$ can be added to the input $u_k$. 
    It is shown in reference~\cite{lewis2010reinforcement} that the solution computed by PI differs from the actual value corresponding
    to the Bellman equation when the probing noise term $n_k$ is added. It is discussed in the same reference that adding the discount factor $\gamma$ reduces this harmful effect of $n_k$.
\end{remark}

\section{The proposed approach}
\label{section: our approach}
This section presents how to perform policy evaluation and policy improvement steps by designing a two-layer QNN.
 According to reference~\cite{bradtke1994adaptive}, the Q-function is a quadratic form
     \begin{gather} \label{Q_function_linear_systems}
         Q^{\pi}(x_k,u_k) = 
         \begin{bmatrix} x_k \\ u_k  \end{bmatrix}^T
         H^\pi
         \begin{bmatrix} x_k \\ u_k  \end{bmatrix} 
     \end{gather}
\noindent where $H^\pi \in \mathbb{R}^{(n_x + n_u) \times (n_x + n_u)}$. As a result, a QNN with coefficients $b=c=0$, $a=1$  is the perfect candidate to approximate the Q-function. 
The policy evaluation step then obtains $H^\pi$ for the stabilizing policy $\pi(.)$ by solving,
     %\begin{multline}
     %\label{eq: policy_evaluation_vanilla_LQR}
         %\begin{bmatrix} x_k \\ u_k   \end{bmatrix}^T
       %   {H}^\pi
    %\begin{bmatrix} x_k \\ u_k    \end{bmatrix}
       %  = x_k^T Q x_k  +  u_k ^T R u_k  +  \\ \gamma
       % \begin{bmatrix} x_{k+1} \\ u_{k+1}   \end{bmatrix}^T
      %    {H}^\pi
    %\begin{bmatrix} x_{k+1} \\ u_{k+1}   \end{bmatrix}
  % \end{multline} 
    \begin{multline}
     \label{eq: policy_evaluation_vanilla_LQR_0}
         \begin{bmatrix} x_k \\ u_k   \end{bmatrix}^T
          {H}^\pi
    \begin{bmatrix} x_k \\ u_k    \end{bmatrix}
         = x_k^T Q x_k  +  u_k ^T R u_k  +  \\ \gamma
         \begin{bmatrix} x_{k+1} \\ u_{k+1}   \end{bmatrix}^T
          {H}^\pi
    \begin{bmatrix} x_{k+1} \\ u_{k+1}   \end{bmatrix}
   \end{multline}

\begin{remark}
%    To solve equation~\eqref{eq: iterative Bellman LQR} for $H^\pi_{i+1}$, one needs at least data points.

         Note that~\eqref{eq: policy_evaluation_vanilla_LQR} is a scalar equation, $H^\pi$ is symmetric, and 
        \begin{gather}
            \begin{bmatrix}
                x_k \\ u_k
            \end{bmatrix} \in \mathbb{R}^{n_x + n_u}
        \end{gather}
 Therefore, the matrix $H^\pi$ has  $M = \frac{ (n_x + n_u) (n_x + n_u + 1)}{2}$ unknown independent elements and $ N \geq M$  data samples are needed to obtain $H^\pi$ from equation~\eqref{eq: policy_evaluation_vanilla_LQR}. 
\end{remark}
We propose to obtain $H^\pi$ in the policy evaluation step by employing the iterative equation~\eqref{eq: policy_evaluation_vanilla_LQR}. The proof of convergence is presented in Lemma~\ref{lemma: policy_evaluation_LQR}.
    \begin{lemma}
    \label{lemma: policy_evaluation_LQR}
Assume that the unknown system (\ref{linear_dynamics}) is controllable.
Then the solution of the iterative equation \eqref{eq: policy_evaluation_vanilla_LQR} starting with any initial value $H^\pi_0$ converges to $H^\pi$ provided the initial policy $\pi(.)$ is stabilizing.
       \begin{multline}
\label{eq: policy_evaluation_vanilla_LQR}
      %  \label{eq: iterative Bellman LQR}
              \begin{bmatrix} x_k \\ u_k  \end{bmatrix}^T
          {H}^\pi_{i+1}
    \begin{bmatrix} x_k \\ u_k   \end{bmatrix}
         = x_k^T Q x_k  + u_k^T R u_k  + \\ \gamma
         \begin{bmatrix} x_{k+1} \\ u_{k+1}    \end{bmatrix}^T
          {H}^\pi_i
    \begin{bmatrix} x_{k+1} \\ u_{k+1}   \end{bmatrix}
     \end{multline}
    \end{lemma}
    \noindent \begin{proof}
    \noindent Applying  equation~\eqref{eq: policy_evaluation_vanilla_LQR}
%{eq: iterative Bellman LQR} 
recursively yields
        \begin{multline}
                \begin{bmatrix} x_k \\ u_k  \end{bmatrix}^T {H}_{i+1}^\pi \begin{bmatrix} x_k \\ u_k   \end{bmatrix} 
             =  \sum^{i}_{j=0}  \gamma^j c(x_{k+j}, u_{k+j}) + \\ \gamma^{i+1}
             \begin{bmatrix} x_{k+i+1} \\ u_{k+i+1} \end{bmatrix}^T
              {H}_{0}^\pi
             \begin{bmatrix}  x_{k+i+1} \\ u_{k+i+1}  \end{bmatrix}  
        \end{multline}
   %     \begin{equation} 
   %     \begin{split}
   %         \begin{bmatrix} x_k \\ u_k^\pi  \end{bmatrix}^T {H}_{i+1}^\pi \begin{bmatrix} x_k \\ u_k^\pi   \end{bmatrix} = & c(x_k,u_k^\pi) + \gamma
   %      \begin{bmatrix} x_{k+1} \\ u_{k+1}^\pi  \end{bmatrix}^T
   %       {H}_{i}^\pi
   %      \begin{bmatrix} x_{k+1} \\ u_{k+1}^\pi  \end{bmatrix}  \\
   %      = & c(x_k,u_k^\pi) + \gamma c(x_{k+1}, u_{k+1}^\pi) + \gamma^2
   %      \begin{bmatrix} x_{k+2} \\ u_{k+2}^\pi  \end{bmatrix}^T
   %       {H}_{i-1}^\pi
   %      \begin{bmatrix} x_{k+2} \\ u_{k+2}^\pi  \end{bmatrix}  \\
   %      = &\sum^{i}_{j=0}  \gamma^j c(x_{k+j}, u_{k+j}^\pi) + \gamma^{i+1}
   %      \begin{bmatrix} x_{k+i+1} \\ u_{k+i+1}^\pi  \end{bmatrix}^T
   %       {H}_{0}^\pi
   %      \begin{bmatrix}  x_{k+i+1} \\ u_{k+i+1}^\pi  \end{bmatrix}  \\
   %      \end{split}
   % \end{equation}
It should be noted that the policy $\pi(.)$ is a stabilizing policy, therefore $\lim_{i\to\infty}x_{k+i+1}=0$ and $\lim_{i\to\infty}u_{k+i+1}=0$ for all $k$. Consequently, for any $H^\pi_0$,   
\begin{gather} \label{eq: one use}
         \lim_{i\to\infty}\gamma^{i+1}
         \begin{bmatrix} x_{k+i+1} \\ u_{k+i+1}^\pi  \end{bmatrix}^T
          {H}_{0}^\pi
         \begin{bmatrix}  x_{k+i+1} \\ u_{k+i+1}^\pi  \end{bmatrix}= 0
\end{gather}

\noindent and therefore
    \begin{eqnarray}\label{mainresult}
        %\begin{split}
           \lim_{i\to\infty} \begin{bmatrix} x_k \\ u_k  \end{bmatrix}^T {H}_{i+1}^\pi \begin{bmatrix} x_k \\ u_k   \end{bmatrix} = 
            \sum^{\infty}_{j=0}  \gamma^j c(x_{k+j}, u_{k+j})=\nonumber\\
         %\end{split}
	=c(x_k,u_k)+\gamma V^\pi(x_{k+1})
    \end{eqnarray}
\iffalse
\noexpand Given a sufficiently large value of $i$, since the value function is finite for a stabilizing policy, the following relation holds,
    \begin{equation} 
    \label{eq: LQR goos}
        \begin{split}
            \begin{bmatrix} x_k \\ u_k  \end{bmatrix}^T {H}_{i+1}^\pi \begin{bmatrix} x_k \\ u_k   \end{bmatrix} \approx &  \sum^{i}_{j=0}  \gamma^j c(x_{k+j}, u_{k+j}) \\  \approx & \sum^{\infty}_{j=0}  \gamma^j c(x_{k+j}, u_{k+j})  \\ 
         \end{split}
    \end{equation}
 \noindent Thus, ${H}_{i}^\pi$ provides an approximation of the value-function for sufficiently large $i$, and 
\fi
Therefore, from equations (\ref{eqn: Q-function_def}), (\ref{Q_function_linear_systems}), and (\ref{mainresult}), $H_i^\pi\to H^\pi$ when $i\to\infty$.
\end{proof}

\begin{remark}
%Note that the most challenging scenario of convergence arises when $\gamma = 1$, given that all future local costs have equal weight. By adopting a reduced value of $\gamma$, we effectively assign less significance to future local costs. 
\iffalse
Consequently, equation \eqref{eq: LQR goos} holds for smaller values of $i$, ultimately leading to less iterations to satisfy
\fi
In practice the condition
    \begin{equation}
        \lVert H^\pi_i - H^\pi_{i-1}  \rVert < \epsilon,
    \end{equation}
is used as the stopping criterium of the algorithm.
\end{remark}
Due to the implications of Lemma~\ref{lemma: policy_evaluation_LQR}, the problem of policy evaluation transforms into the task of computing $H^\pi_{i+1}$ from equation~\eqref{eq: policy_evaluation_vanilla_LQR} given  $H^\pi_{i}$. This can be done by training a two-layer QNN.
Since it is assumed that one has access to the state $x_k$, one can calculate $ \mathcal{Y}_k$ defined as 
\begin{gather}
    \mathcal{Y}_k = x_k^T Q x_k + u_k^T R u_k + 
          \gamma
          \begin{bmatrix} x_{k+1} \\ u_{k+1}  \end{bmatrix}^T
          {H}_{i}^\pi
         \begin{bmatrix} x_{k+1} \\ u_{k+1}  \end{bmatrix}.
\end{gather}
\noindent Therefore, from (\ref{eq: policy_evaluation_vanilla_LQR}),
\begin{equation}
    \mathcal{X}_k^T {H}_{i+1}^\pi \mathcal{X}_k = \mathcal{Y}_k
\end{equation}
where $\mathcal{X}_k^T = \begin{bmatrix} x_k^T & u_{k}^T \end{bmatrix}$.
Thus, $H^\pi_{i+1}$ can be obtained from the training of a QNN as the solution of the convex optimization~\eqref{convex_formula} using a set of input data points $\mathcal{X}_k$ and the corresponding labels $\mathcal{Y}_k$, as well as the coeffficients $a=1, \: b=0, \: c=0$.
     
\subsection{Policy improvement: }
     \noindent In this section, the policy improvement step is addressed for the stabilizing policy $\pi(\cdot)$ using $H^\pi$ from the policy evaluation. We first partition  $H^\pi$ as
    \begin{gather}
         \begin{bmatrix} x_k \\ u_k \end{bmatrix}^T
         {H}^\pi
         \begin{bmatrix} x_k \\ u_k  \end{bmatrix} = 
         \begin{bmatrix} x_k \\ u_k  \\ \end{bmatrix}^T
         \begin{bmatrix} 
             {H}_{xx}^\pi & {H}_{xu}^\pi  \\ 
            ({H}_{xu}^\pi)^T & {H}_{uu}^\pi  
         \end{bmatrix}
         \begin{bmatrix} x_k \\ u_k  \end{bmatrix}
    \end{gather}

     \begin{lemma}
      The policy improvement for the stabilizing policy $\pi(.)$ is given by
     \begin{gather}
              \pi^\prime (x_k) = - ({H}_{uu}^\pi)^{-1}({H}_{xu}^\pi)^Tx_k \label{eq: policy improvement}
     \end{gather}
    \end{lemma}
    
    \begin{proof}
The improved policy $\pi^\prime (.)$ is obtained by
    \begin{gather} 
         \pi^\prime (x_k) = \arg \min_{u_k} Q^{\pi}(x_k,u_k) 
         \label{eq: policy improvement vanila_2}
    \end{gather} 
     \noindent Since  $Q^{\pi}(x_k,u_k )$ is a quadratic form, the necessary and sufficient conditions of optimality are
    \begin{gather} 
         \nonumber  \frac{\partial Q^{\pi}(x_k,u_k )}{\partial u_k } = 0 \\ 
         \frac{\partial^2 {Q}^{\pi}(x_k,u_k )}{\partial u_k^2}  >0 
          \label{eq: solve_policy_imp_by_derivative}
    \end{gather} 
    
    \noindent The Q-function  can be written as 
         \begin{gather}
         \nonumber Q^{\pi}(x_k , u_k ) = 
         \begin{bmatrix} x_k \\ u_k  \\ \end{bmatrix}^T
         \begin{bmatrix} 
             {H}_{xx}^\pi & {H}_{xu}^\pi  \\ 
            ({H}_{xu}^\pi)^T & {H}_{uu}^\pi  
         \end{bmatrix}
         \begin{bmatrix} x_k \\ u_k  \end{bmatrix} = \\
         x_k^T    H_{xx}^\pi   x_k    +   x_k^T  H_{xu}^\pi    u_k +
         u_k^T  (H_{xu}^\pi)^T   x_k   +   u_k^T  H_{uu}^\pi   u_k  \label{QFA_open_form_for_b=c=0}
     \end{gather}
     
     \noindent Therefore, the solution to the first constraint yields
    \begin{gather}
         \nonumber \frac{\partial Q^{\pi}(x_k,u_k )}{\partial u_k } = 0\Leftrightarrow
         {H}_{uu}^\pi u_k  + ({H}_{xu}^\pi)^Tx_k  = 0 \Leftrightarrow \\
         u_k  = - ({H}_{uu}^\pi)^{-1}({H}_{xu}^\pi)^Tx_k \label{eq: policy improvement}
     \end{gather}

     \noindent  Note that the matrix $H^\pi$ is positive definite. Therefore, all matrices on the main diagonal of $H^\pi$, including $H_{uu}^\pi$, are also positive definite. As a result, the inverse of $H_{uu}^\pi$ does exist.
 
    \noindent The second constraint in~\eqref{eq: solve_policy_imp_by_derivative} is thus satisfied since $H_{uu}^\pi >0$.
Thus, the policy obtained in~\eqref{eq: policy improvement} is  the unique minimizer $\pi^\prime (x_k)$.
\iffalse
    \begin{gather}
          \pi^\prime (x_k) =  - (H_{uu}^\pi)^{-1}  (H_{xu}^\pi)^T  x_k  
     \end{gather}     
\fi      
\end{proof} 
    %\noindent \textbf{Remark 6:} If the system model is known, one can add constraints to~\eqref{convex_formula} to make sure that the system is stable even with the added noise and discount factor.
The complete procedure for solving the LQR is presented in Algorithm~\ref{alg: PI_LQR}.

          \begin{algorithm}[t]
        \caption{Solving LQR with QNN and Q-learning } \label{alg: PI_LQR}
        \begin{algorithmic} 
            \State Choose $\epsilon$, $N$, $\gamma$.
            \State Select the initial stabilizing policy $ \pi_0$. Then, for $j=0,1,2,\hdots$ perform policy evaluation and policy improvement steps until  convergence 
            \State \textbf{Policy Evaluation:}
            \State $\: \: \: \: \:$ $i \gets 0$   
            \State $\: \: \: \: \:$ Choose a random $H^{\pi_j}_0$
            \State $\: \: \: \: \:$ \textbf{repeat}  
             
            \State $\: \: \: \: \:$ $\: \: \: \: \: $  Train the QNN by $N$ data samples with $\mathcal{X}_k$ as the  
            \State $\: \: \: \: \:$ $\: \: \: \: \:  \: \: \: \: \: \:$ inputs  and $\mathcal{Y}_{k}$ as the output labels.
            \State $\: \: \: \: \:$ $\: \: \: \: \:$ Obtain the input-output mapping as 
                \begin{equation}
                    \mathcal{X}_k^T H \mathcal{X}_k = \mathcal{\hat{Y}}_{k}
                \end{equation}
                \State $\: \: \: \: \:$  $\: \: \: \: \:$ $i \gets i+1 $
            \State $\: \: \: \: \:$ $\: \: \: \: \:$ $H^{\pi_{j}}_i \gets H$
 
            \State $\: \: \: \: \:$ \textbf{Until}  $||H^{\pi_j}_i - H^{\pi_j}_{i-1} ||< \epsilon$
            \State $\: \: \: \: \:$ $H^{\pi_j} \gets H^{\pi_j}_i$
        \vspace{5pt}
            \State \textbf{Policy improvement:}
            \State $\: \: \: \: \:$ Obtain $\pi_{j+1}    $  such that   
            \begin{equation}
                \begin{split} 
                    \: \: \: \: \:  \pi_{j+1} =  - (H_{uu}^{\pi_j})^{-1}  (H_{xu}^{\pi_j})^T  x_k                 
                \end{split}
            \end{equation} 
        \end{algorithmic}
    \end{algorithm}

    \section{Simulations}
    \label{section: LQR implementation}
    
         This sectiom considers a quadrotor flying at a constant altitude movimg in the $x_1$ direction. It is shown that the control policy converges to the LQR solved by conventional methods. 
    The quadrotor is modeled by the state-space
\begin{gather}
      \begin{bmatrix}
        \dot{x}_1  \\ 
        \dot{x}_2  \\
        \dot{x}_3\\
        \dot{x}4
      \end{bmatrix} =
      \begin{bmatrix}
          0 & 1 & 0 & 0   \\ 
          0 & -0.1 & 10 & 0\\
          0 & 0 & 0 & 1 \\
          0 & 0 & 0 & 0
      \end{bmatrix}
      \begin{bmatrix}
          x_1\\ 
          x_2\\
          x_3\\
          x_4
      \end{bmatrix}
      +
      \begin{bmatrix}
          0 \\ 
          0\\
          0\\
          4.35
      \end{bmatrix}u
\end{gather} 
 where $x_1$ is the horizontal position, $x_2$ is the speed, $x_3$ is the pitch angle, $x_4$ is the pitch rate
and the initial state is
\begin{equation}
    \begin{bmatrix}
        x_1(0) \\
        x_2(0) \\
        x_3(0) \\
        x_4(0)
    \end{bmatrix} = 
        \begin{bmatrix}
        -10 \\
        0   \\
        0   \\
        0
    \end{bmatrix} 
\end{equation}
Assume that the sampling time is $T=0.1$ seconds. The discretized model~\eqref{eq: linear_quadrotor_tustin} is obtained using Tustin's method in order to compare the result with the discrete-time LQR.
\begin{gather}
      \begin{bmatrix}
        x_{1,k+1} \\ 
        x_{2,k+1} \\
        x_{3,k+1} \\
        x_{4,k+1} 
      \end{bmatrix} =
      \begin{bmatrix}
            1  &  0.1  &  0.05  &    0.003 \\
                 0  &  0.99  &  0.99  &  0.05 \\
                 0  &       0  &  1 &  0.1 \\
                 0  &       0  &       0  &  1
      \end{bmatrix}
      \begin{bmatrix}
        x_{1,k } \\ 
        x_{2,k } \\
        x_{3,k } \\
        x_{4,k } 
      \end{bmatrix}
      +
      \begin{bmatrix}
            0.001 \\
            0.011 \\
            0.022 \\
            0.435
      \end{bmatrix}u_k  \label{eq: linear_quadrotor_tustin}
\end{gather}

\noindent Rewrite the control policy $\pi_j(.)$ as:  
    \begin{gather} 
            \pi_j(x_k)=
            \begin{bmatrix}
                k_1^{\pi_j} & k_2^{\pi_j} & k_3^{\pi_j} & k_4^{\pi_j} 
            \end{bmatrix}
            \begin{bmatrix}
                x_{1,k} \\ x_{2,k} \\ x_{3,k} \\ x_{4,k} 
            \end{bmatrix}  
            \label{eq: shand}    
    \end{gather}
 In this example, we choose  the following parameter values:
    \begin{gather}
       \nonumber  R = 100,\: N=100, \: \gamma=1, \:\beta=0.005 \\
          Q = \begin{bmatrix}
          0.01 & 0 &  0  & 0 \\
         0   &  1  & 0  & 0 \\
         0   &  0 &  1 &  0 \\
         0   &  0 &  0 &  10
        \end{bmatrix}
    \end{gather} 
 
  \noindent The optimal LQR controller is 
     \begin{gather} 
            \pi^*(x_k)=-
            \begin{bmatrix}
                  0.046 &  0.464 &  4.347  & 2.014
            \end{bmatrix}            
            \begin{bmatrix}
                x_{1,k} \\ x_{2,k} \\ x_{3,k} \\ x_{4,k}
            \end{bmatrix}    
    \end{gather}
We now pretend that we do not know the model of the quadrotor and run the Algorithm~\ref{alg: PI_LQR}.
Fig~\ref{fig: Convergence_quadrotor} demonstrates the convergence of $k_1^{\pi_j}$, $k_2^{\pi_j}$, $k_3^{\pi_j}$, $k_4^{\pi_j}$ to their optimal values. To illustrate this convergence, we conducted five simulations using random initial stabilizing policies. The random initial policies are given in the Table~\ref{table: LQR_quadrotor}.
\begin{figure}[t]
    \centering
        \includegraphics[width=3.8in]{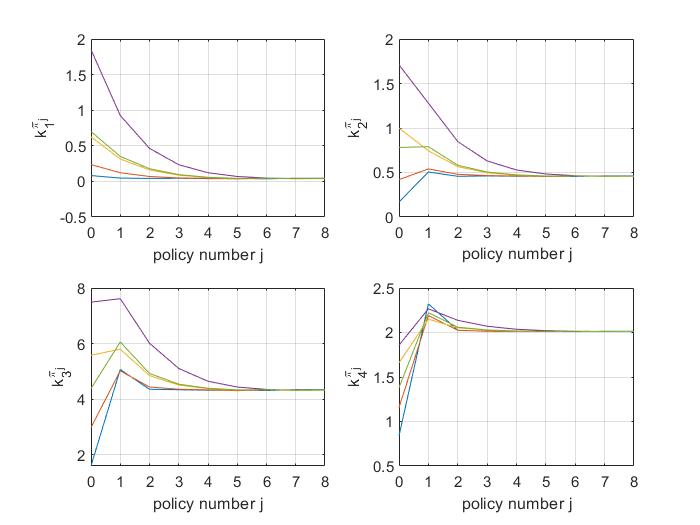}
        \label{fig:first}
    \caption{Convergence of the policy iteration to LQR controller}
    \label{fig: Convergence_quadrotor}
\end{figure}
The trajectory of the quadrotor's position and its speed over time using the optimal policy are depicted in Fig.~\ref{fig: UAV_trajectory} where one sees the convergence to the origin.
\begin{figure}[htbp]
    \centering
        \includegraphics[width=3.8 in]{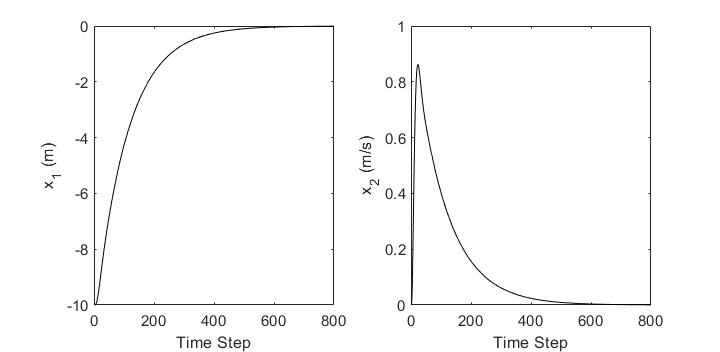}
        \label{fig:first}
    \caption{Position and speed trajectory using the optimal policy.}
    \label{fig: UAV_trajectory}
\end{figure}

     %The simulation starts with the initial state $x_0 = \begin{bmatrix}
     %      0.4 & 0
     %\end{bmatrix}^T$, and an initial admissible policy $\pi_0 (.)$.

    %  \noindent To evaluate policy $\pi_i(.)$
    % and calculate $\pi_{i+1} (.)$, we start from the initial state $x_0 = \begin{bmatrix}
    %       0.4 & 0
    %  \end{bmatrix}^T$ and follow the policy $\pi_i (.)$ until 50 data points are gathered to evaluate and improve the policy.
    %     We continue to improve the policy until  $k_1(i)$ and $k_2(i)$ converge to $k_1^*$ and $k_2^*$, respectively. The cost-to-go from the initial state $x_0$ also decreases and reaches a minimum as policy number increases. We choose two different $c(x_k,u_k)$ for the pendulum example.

         \begin{table}[t]
            \centering
            \begin{tabular}{|c|c|}
                  \hline
                  Simulation number & $K^{\pi_0}$ \\
                  \hline
                  Simulation 1 & $\begin{pmatrix}   0.082  &  0.169  &  1.592  &  0.838   \end{pmatrix}$ \\
                  Simulation 2 & $\begin{pmatrix}   0.236  &  0.420  &  2.978  &  1.156   \end{pmatrix}$ \\
                  Simulation 3 & $\begin{pmatrix}   0.626  &  0.998  &  5.578  &  1.660    \end{pmatrix}$ \\
                  Simulation 4 & $\begin{pmatrix}   1.851  &  1.709  & 7.491   &  1.858  \end{pmatrix}$ \\
                  Simulation 5 & $\begin{pmatrix}   0.701 &   0.781  &  4.380  &   1.379   \end{pmatrix}$ \\
                  \hline
            \end{tabular}
            \caption{The random initial stabilizing policies }
            \label{table: LQR_quadrotor}
    \end{table}

\bibliographystyle{unsrt} % We choose the "plain" reference style
\bibliography{citation} % Entries are in the refs.bib file

%\bibliographystyle{ieeetr}
%\bibliography{citation} 

\end{document}